\documentclass[twocolumn]{svjour3}

\usepackage{bm}

\newcommand\mymatrix[1]{\bm{\mathrm{#1}}}

\usepackage{graphicx}
\graphicspath{{figures-pdf/}}
\usepackage[retainorgcmds]{IEEEtrantools}
\newlength\imagewidth
\usepackage{color}
\definecolor{dgreen}{rgb}{0,.6,0}
\usepackage[normalem]{ulem}
\usepackage{cite}
\usepackage{amsmath}

\usepackage[nolabel]{fnbreak}
\journalname{Nonlinear Dynamics}

\usepackage{lineno}

\newcommand\Xkey{\textit{Xkey}}

\newcommand\CKS{\textit{CKS}}
\newcommand\key{\textit{key}}
\newcommand\HD{\textup{HD}}
\newcommand\VD{\textup{VD}}

\begin{document}

\title{Cryptanalyzing image encryption scheme using chaotic logistic map}

\author{Chengqing Li \and Tao Xie \and Qi Liu \and Ge Cheng}

\authorrunning{C. Li et al.}

\institute{
Chengqing Li, Tao Xie, Qi Liu\at
              MOE Key Laboratory of Intelligent Computing and Information Processing,
              College of Information Engineering,
              Xiangtan University, Xiangtan 411105, Hunan, China \\
              Tel.: +86-731-52639779\\
              Fax: +86-731-58292217\\
              \email{Chengqingg@gmail.com}\\
Ge Cheng \at
School of Mathematics and Computational Science,
Xiangtan University, Xiangtan 411105, Hunan, China
}

\date{Received: Jun 7 2014}

\maketitle

\begin{abstract}
Chaotic behavior arising from very simple non-linear dynamical equation of logistic map makes
it was used often in designing chaotic image encryption schemes. However, some properties of chaotic maps
can also facilitate cryptanalysis especially when they are implemented in digital domain. Utilizing stable distribution of the chaotic states
generated by iterating the logistic map, this paper present a typical example to show insecurity
of an image encryption scheme using chaotic logistic map. This work will push encryption and chaos
be combined in a more effective way.


\end{abstract}

\section{Introduction}

With the rapid development of information transmission technology and popularization of
multimedia capture device, multimedia data are transmitted over all kinds of wired/wireless networks more
and more frequently. Consequently, the security of multimedia data
becomes a serious concern of many people. However, the traditional
text encryption schemes cannot protect multimedia data efficiently, mainly due to the
big differences between textual and multimedia data, such as strong redundancy between
neighboring elements of uncompressed multimedia data, bulky size of multimedia data, and some special
requirements of the whole multimedia processing system. This challenge stirs
the design of special multimedia encryption schemes to become a hot
research topic in multimedia processing area in the past
two decades.

Because there are some subtle similarities between chaos and
cryptography, chaos were used to design potential secure and efficient schemes in all kinds of applications
in cryptography. Due to easy presentation of image datum, most cryptographic schemes consider image data as
operation object \cite{YaobinMao:CSF2004,Pareek:ImageEncrypt:IVC2006,Tong:ImageCipher:IVC07}. According to the record of \textit{Web of Science},
about two thousands of articles on chaos-based cryptology were published in the past two decades. Unfortunately, many of them have been found
to have security problems of different degree from the viewpoint of modern cryptology
\cite{David:AttackingChaos08,Goce:cryptanalysis:TM08,Li:BreakImageCipher:IVC09,Alvarez:PLA2003,Li:AttackingCNSNS2008,Li:AttackingIVC2009,Zhang:perceptron:ND12}. In general, cryptanalysis (breaking or reporting some security defects) of
a given encryption scheme requires rigorous theoretical support, so only about one tenth of the proposed chaos-based encryption schemes were
reported being cryptanalyzed till now. At the same time, the conventional cryptanalysis skills cannot be directly used to break the chaos-based encryption schemes
due to some reasons, such as difference of essential structure, and some special security defects caused by the properties of the adopted chaos system \cite{SJLi:degradation:IJBC05,ZhouJT:Coder:TCSI11,FChen:PeriodArnold:TIT12}. Therefore, short of security scrutiny
of the proposed chaos-based encryption schemes has become bottleneck of progress of chaos-based cryptology. Some general rules about evaluating
security of chaos-based encryption schemes are concluded in \cite{AlvarezLi:Rules:IJBC2006,Li:ChaosImageVideoEncryption:Handbook2004}.

Logistic map is a discrete quadratic recurrence form of the logistic equation,
a model of population growth first published by P. Verhulst in 1845. The application of logistic map in cryptology
can be traced back to John von Neumann's suggestion on utilizing it as a random number generator in 1947 \cite{Neumann:logistic:BAMS47}.
Due to simple form and relatively complex dynamical properties of logistic map, it was extensively used to design encryption schemes or pseudo-random
number generator \cite{Baptista:Lotistic:PLA98,Phatak:LogisticRNG:PRE95,Kocarev:Logistic:PLA01}. Even in \textit{Web of Science}, one can find
about two hundred records on application of logistic map in cryptography published between 1998 and 2013.
A few papers reported some security deficiencies caused by logistic map \cite{Li:AttackingRCES2008,Persohn:AnalyzeLogistic:CSF12}.

In \cite{Pareek:CNSNS2009}, an image encryption scheme based on the logistic and standard maps was proposed, where the two maps
are iterated to generate some pseudo-random number sequences (PRNS) controlling twice exclusive OR (XOR) operations. In \cite{Rhouma:BreakPareek:CNSNS10},
it is reported that an equivalent key of the scheme can be obtained from only one known/chosen plain-image and the corresponding cipher-image. However,
the equivalent key can only be used to decrypt other cipher-images of smaller or the same size of the known/chosen plain-image. Based on a dynamical property of logistic map on stable distribution of its chaotic states, the present paper re-evaluates the security of the
scheme, and finds that the scope of all the sub-keys can be narrowed much by an efficient brute-force search.

The rest of this paper is organized as follows. Section~\ref{sec:encryptscheme} introduces the image encryption
scheme under study briefly. Our cryptanalytic results are presented
in Sec.~\ref{sec:cryptanalysis} in detail. The last section
concludes the paper.

\section{The image encryption scheme under study}
\label{sec:encryptscheme}

The plaintext encrypted by the image encryption scheme under study
is a RGB true-color image of size $M\times N$ (height$\times$width),
which can be denoted by an $M\times N$ matrix of 3-tuple pixel
values $\mymatrix{I}=\{(R(i,j), G(i,j), $ $B(i,j))\}_{0\leq i\leq M-1 \atop
0\leq j\leq N-1}$. Denoting the cipher image by $\mymatrix{I}'=\{(R'(i,j),
G'(i,j),$ $ B'(i,j))\}_{0\leq i\leq M-1\atop 0\leq j\leq N-1}$, the
image encryption scheme can be described as follows\footnote{To make
the presentation more concise and complete, some notations in the
original paper are modified provided that its essential form kept unchanged.}.

\begin{itemize}
\item
\textit{Secret key}: three floating-point numbers $x_0$, $y_0$, $K$,
and one integer $L$, where $x_0$, $y_0\in(0, 2\pi)$, $K>18$,
$100<L<1100$.

\item
\textit{Initialization}: prepare data for encryption/decryption by
performing the following steps.
\begin{itemize}
\item[a)]
Generate four XORing keys as follows: $\textit{Xkey}(0)=\lfloor
256x_0/(2\pi)\rfloor$, $\textit{Xkey}(1)=\lfloor
256y_0/(2\pi)\rfloor$, $\textit{Xkey}(2)$ $=\lfloor K\bmod 256
\rfloor$, $\textit{Xkey}(3)=(L \bmod 256)$. Then, generate a pseudo-image
$\mymatrix{I}_{\Xkey}=\{(R_{\Xkey}(i,j),$ $G_{\Xkey}(i,j),$ $
B_{\Xkey}(i,j))\}_{0\leq i\leq H-1 \atop 0\leq j\leq W-1}$ by
assigning an $M\times N$ matrix with the four XORing keys repeatedly:\\
$R_{\Xkey}(i,j)=\Xkey(3k\bmod 4)$,\\
$G_{\Xkey}(i,j)=\Xkey((3k+1)\bmod 4)$,\\
$B_{\Xkey}(i,j)=\Xkey((3k+2)\bmod 4)$, where $k=iN+j$.

\item[b)]
Iterate the standard map
\begin{equation*}
\left\{
\begin{array}{lcl}
x   &=&   (x+K\sin(y))\bmod(2\pi),\\
y   &=&  (y+x+K\sin(y))\bmod(2\pi),
\end{array}
\right.
\end{equation*}
from the initial conditions $(x_0,y_0)$ $L$ times to obtain a new chaotic state
$(x_0', y_0')$. Then, further iterate it $MN$ more times to get
$MN$ chaotic states $\{(x_i,y_i)\}_{i=1}^{MN}$.

\item[c)]
Iterate the logistic map
\begin{equation}
\label{eq:logistic} f(x)=4x(1-x)
\end{equation}
from the initial condition $z_0=((x_0'+y_0')\bmod 1)$ for $L$ times to get a new
initial condition $z_0'$. Then, further iterate it for $MN$ times to
get $MN$ chaotic states $\{z_i\}_{i=1}^{MN}$.

\item[d)]
Generate a chaotic key stream (CKS) image $\mymatrix{I}_{\CKS}=\{(\textit{CKSR}(i,j),$ $\textit{CKSG}(i,j),$
$\textit{CKSB}(i,j))\}_{0\leq i\leq M-1 \atop 0\leq j\leq N-1}$ as
follows:
$\textit{CKSR}(i,j)=\left\lfloor
256x_k/(2\pi)\right\rfloor$, $\textit{CKSG}(i,j)$ =$\left\lfloor
256y_k/(2\pi)\right\rfloor$ and
\begin{equation}\label{eq:CKSB}
\textit{CKSB}(i,j)=\tilde{z}_k=\left\lfloor
256z_k\right\rfloor,
\end{equation}
where $k=iN+j+1$.
\end{itemize}

\item
\textit{Encryption procedure}: a simple concatenation of the
following four encryption operations.

\begin{itemize}
\item
\textit{Confusion I}: masking the plain pixel values by the four
XORing keys $\{\textit{Xkey}(i)\}_{i=0}^3$.
For $k=0,\ldots,MN-1$, set
\begin{eqnarray*}
R^{\star}(i,j) & = & R(i,j)\oplus R_{\Xkey}(i,j),\\
G^{\star}(i,j) & = & G(i,j)\oplus G_{\Xkey}(i,j),\\
B^{\star}(i,j) & = & B(i,j)\oplus B_{\Xkey}(i,j),
\end{eqnarray*}
where $i=\lfloor k/N\rfloor$, $j=(k\bmod N)$.

\item
\textit{Diffusion I}: scanning all pixel values from the first one
row by row (from top to bottom), and masking each pixel (except for
the first scanned pixel) by its predecessor in the scan.

Set $R^{*}(0, 0)=R^{\star}(0, 0)$, $G^{*}(0, 0)=G^{\star}(0, 0)$,
$B^{*}(0, 0)=B^{\star}(0, 0)$. For $k=1,\ldots,MN-1$, set
\begin{eqnarray*}
R^{*}(i, j) & = &  R^{\star}(i, j) \oplus R^{*}(i', j'),\\
G^{*}(i, j) & = &  G^{\star}(i, j) \oplus G^{*}(i', j'),\\
B^{*}(i, j) & = &  B^{\star}(i, j) \oplus B^{*}(i', j'),
\end{eqnarray*}
where $i=\lfloor k/N\rfloor$, $j=(k\bmod N)$, $i'=\lfloor
(k-1)/N\rfloor$ and $j'=((k-1)\bmod N)$.

\item
\textit{Diffusion II}: scanning all pixel values from the last one
column by column (from right to left), and masking each pixel
(except for the first scanned pixel) by its predecessor in the scan.

Set $R^{**}(M-1, N-1)=R^{*}(M-1, N-1)$, $G^{**}(M-1, N-1)=G^{*}(M-1,
N-1)$, $B^{**}(M-1, N-1)=B^{*}(M-1, N-1)$. For $k=MN-2,\ldots,0$, set
\begin{eqnarray*}
R^{**}(i, j)  =  R^{*}(i, j)\oplus G^{**}(i',j')\oplus B^{**}(i',j'),\\
G^{**}(i, j)  =  G^{*}(i, j)\oplus B^{**}(i',j')\oplus R^{**}(i',j'),\\
B^{**}(i, j)  =  B^{*}(i, j)\oplus R^{**}(i',j')\oplus
G^{**}(i', j'),
\end{eqnarray*}
where $i=(k\bmod M)$, $j=\lfloor k/M \rfloor$, $i'=((k+1)\bmod M)$,
$j'=\lfloor (k+1)/M \rfloor$.

\item
\textit{Confusion II}: masking the pixel values with the CKS image
pixel by pixel.
For $k=0,\ldots,MN-1$, set
\begin{eqnarray*}
R'(i,j) & = & R^{**}(i,j)\oplus \textit{CKSR}(i,j),\\
G'(i,j) & = & G^{**}(i,j)\oplus \textit{CKSG}(i,j),\\
B'(i,j) & = & B^{**}(i,j)\oplus \textit{CKSB}(i,j).
\end{eqnarray*}
where $i=\lfloor k/N\rfloor$, $j=(k\bmod N)$.
\end{itemize}

\item
\textit{Decryption procedure} is the simple reversion of the above
encryption procedure.
\end{itemize}

\section{Cryptanalysis}
\label{sec:cryptanalysis}

In \cite[Sec.~4.2]{Rhouma:BreakPareek:CNSNS10}, it has been reported
that an equivalent secret key of the image encryption scheme under study can be reconstructed with only one
pair of known plain-image and the corresponding cipher-image. More rigorous presentation of the attack was represented in \cite[Sec.~3.1]{Li:BreakPareek2:CNSNS11}. The equivalent form of the secret key can only decrypt other cipher-image of smaller or the same size of
the known plain-image. In this section, we first briefly introduce how the equivalent secret key is obtained, then discuss how to derive
further information about secret key utilizing special dynamical properties of the logistic map.

Denoting the horizontal and vertical diffusion processes by HD and
VD, respectively, the image encryption scheme under study can be represented as
\begin{IEEEeqnarray*}{rCl}
\mymatrix{I}' &=& \VD(\HD(\mymatrix{I}\oplus\mymatrix{I}_{\Xkey}))\oplus\mymatrix{I}_{\CKS} \nonumber\\
              &=& \VD(\HD(\mymatrix{I}))\oplus \VD(\HD(\mymatrix{I}_{\Xkey}))\oplus \mymatrix{I}_{\CKS}\\
              &=& \VD(\HD(\mymatrix{I}))\oplus \mymatrix{I}_{\key},
\end{IEEEeqnarray*}
where
\begin{equation}
\mymatrix{I}_{\key}=\mymatrix{I}_{\Xkey}^*\oplus \mymatrix{I}_{\CKS},
\label{eq:equivalentkey}
\end{equation} and
\begin{equation*}
\mymatrix{I}_{\Xkey}^*=\VD(\HD(\mymatrix{I}_{\Xkey}).
\end{equation*}
As both HD and VD are independent on the secret key and $\mymatrix{I}_{\key}$
is dependent on neither the plaintext
$\mymatrix{I}$ nor the ciphertext $\mymatrix{I}'$, one can see that
$\mymatrix{I}_{\key}$ can be used as an equivalent key to decrypt any
ciphertext of size which is smaller than or equal to $M\times N$, encrypted by the same secret key $(x_0,y_0,K,L)$. In the following, we will show that the scope of $x_0$, $y_0$, $K$, and $L$ can be further narrowed under the same condition.

The key idea of the enhanced attack is to search the values of $\{\textit{Xkey}(i)\}_{i=0}^3$ by brute-force
and verify the search with stable self-correlation of the blue channel of $\mymatrix{I}_{\CKS}$, $\{\textit{CKSB}(i,j)\}_{0\leq i\leq M-1 \atop 0\leq j\leq N-1}$,
which is obtained from states of logistic map. Two main points supporting the attack are described as follows.
\begin{itemize}
\item \textit{Deterministic relationship between $\{\textit{Xkey}(i)\}_{i=0}^3$ and $\mymatrix{I}_{\Xkey}^*$}

Obviously, every element of $\mymatrix{I}_{\Xkey}^*$ comes from set $\{\textit{Xkey}(0), \textit{Xkey}(1), \textit{Xkey}(2), \textit{Xkey}(3)$,
$\textit{Xkey}(0)\oplus$ $\textit{Xkey}(1)$, $\textit{Xkey}(0)\oplus\textit{Xkey}(2), \textit{Xkey}(0)\oplus\textit{Xkey}(3), \textit{Xkey}(1)$
$\oplus\textit{Xkey}(2)$,
$\textit{Xkey}(1)\oplus\textit{Xkey}(3), \textit{Xkey}(2)\oplus\textit{Xkey}(3)$, $\textit{Xkey}(0)\oplus\textit{Xkey}(1)\oplus \textit{Xkey}(2)$, $
\textit{Xkey}(0)\oplus\textit{Xkey}(1)\oplus \textit{Xkey}(3)$, $\textit{Xkey}(1)\oplus\textit{Xkey}(2)\oplus \textit{Xkey}(3)$, $\textit{Xkey}(0)\oplus\textit{Xkey}(2)\oplus \textit{Xkey}(3)$,
$\textit{Xkey}(0)\oplus\textit{Xkey}(1)\oplus \textit{Xkey}(2)\oplus \textit{Xkey}(3)\}$. Given $M$, $N$ and $\{\textit{Xkey}(i)\}_{i=0}^3$, $\mymatrix{I}_{\Xkey}^*$ is fixed. Assigning $\{\textit{Xkey}(i)\}_{i=0}^3$
with four different numbers, e.g., $\{1, 2, 4, 9\}$, the construction of $\mymatrix{I}_{\Xkey}^*$ can be known. Scan $\mymatrix{I}_{\Xkey}^*$ in the raster order, and convert it into a one-dimensional sequence. Note that not every element of $\{\textit{Xkey}(i)\}_{i=0}^3$ have independent influence on $\mymatrix{I}_{\Xkey}^*$ when $T=4$, where
$T$ denote period of the one-dimensional version of $\mymatrix{I}_{\Xkey}^*$. For example, the non-periodic component of $\mymatrix{I}_{\Xkey}^*$ is
$\{\textit{Xkey}(2)$, $\textit{Xkey}(2)\oplus \textit{Xkey}(3)$, $\textit{Xkey}(0)\oplus\textit{Xkey}(1)\oplus\textit{Xkey}(3)$, $\textit{Xkey}(0)\oplus\textit{Xkey}(1)\}$ when $M=N=9$,
where different combinations of $\textit{Xkey}(0)$ and $\textit{Xkey}(1)$ may generate the same version of $\mymatrix{I}_{\Xkey}^*$.
Under condition $M=N$, the relationship between $T$ and $N$ is shown in Proposition~\ref{prop:period}. Fortunately, from Proposition~\ref{prop:condition}, one can see that only a
very minor portion of combinations of $M$ and $N$ making not every element of $\{\textit{Xkey}(i)\}_{i=0}^3$ have independent influence on $\mymatrix{I}_{\Xkey}^*$.

\item \textit{Stable self-correlation of $\{\textit{CKSB}(i,j)\}_{0\leq i\leq M-1 \atop 0\leq j\leq N-1}$}

Recall Eq.~(\ref{eq:CKSB}), one can calculate $\Delta_k=\mu_k-4$, where
\begin{equation*}
\mu_k=(\tilde{z}_{k+1}/256)/((\tilde{z}_k/256)\cdot(1-(\tilde{z}_k/256))).
\end{equation*}
As well known, distribution of chaotic trajectories generated by iterating the map Eq.~(\ref{eq:logistic}) is stable. To verify this,
distributions of a great number of chaotic trajectories generated by iterating Logistic map Eq.~(\ref{eq:logistic}) for $10^5$ times under random initial conditions were
studied. All the distributions are quite similar to each other, so only one typical example is shown in Fig.~\ref{fig:distribute4logistic}
(See the invariant density of logistic map with parameter $r=4$ in \cite[Fig.~2]{Oteo:DoublePlogistic:PRE07}). Furthermore, distribution of $\Delta_k$ is stable also under different values of $z_0$. To verify this point, a great number of initial values of $z_0$ are chosen randomly. For each value of $z_0$, the corresponding
sequence $\{\tilde{z}_k\}$ is produced by transforming the corresponding logistical states with Eq.~(\ref{eq:CKSB}). As distributions of $\Delta_k$ are similar to each other, only the distribution of $\Delta_k$ under the initial value used in Fig.~\ref{fig:distribute4logistic} is shown in Fig.~\ref{fig:z0z1}. To show the point more clearly, the distribution of $\Delta_k$ over some given intervals are shown in
Table \ref{table:distributionerror}. In addition,
distribution of $\Delta_k$ is very sensitive to change of $\tilde{z}_k$ or $\tilde{z}_{k+1}$ (See Figs.~\ref{fig:z0z1},~\ref{fig:z0+1z1},~\ref{fig:z0z1+1},~\ref{fig:z0+1z1+1}).
So, this stable self-correlation of $\{\textit{CKSB}(i,j)\}_{0\leq i\leq M-1 \atop 0\leq j\leq N-1}$ can be used to verify the search of $\{\textit{Xkey}(i)\}_{i=0}^3$.
\end{itemize}

\begin{proposition}
When $M=N$, $T$ and $N$ satisfies that
\begin{equation}
 T= \begin{cases}
   4   &   \mbox{if }(N\bmod 2)=1;\\
   2N  &   \mbox{if }(N\bmod 4)=0;\\
   4N  &   \mbox{if }(N\bmod 4)=2.
  \end{cases}
\end{equation}
\label{prop:period}
\end{proposition}
\begin{proof}
This proposition has been verified by computer for $N=3, \cdots, 1204$. We leave
this proposition under the larger scope of $N$ as a conjecture.
\end{proof}

\begin{proposition}
The necessary condition making not every element of $\{\textit{Xkey}(i)\}_{i=0}^3$
have independent influence on $\mymatrix{I}_{\Xkey}^*$
is that
\begin{equation}
|M-N|\leq 1.
\end{equation}
\label{prop:condition}
\end{proposition}
\begin{proof}
This proposition has been verified by computer for $3\leq M, N\leq 1204$. We leave
this proposition under the larger scope of $M$ and $N$ as a conjecture.
\end{proof}

\begin{figure}[!htb]
\centering
\begin{minipage}[t]{\imagewidth}
\centering
\includegraphics[width=\imagewidth]{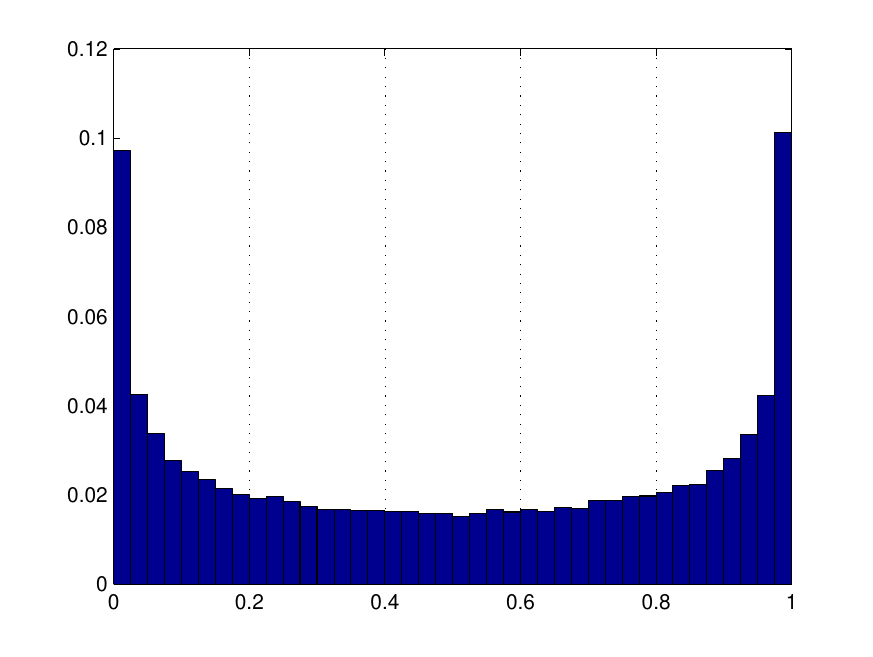}
\end{minipage}
\caption{Distribution of trajectory of the map Eq.~(\ref{eq:logistic}) under initial state $z_0=0.226$.}
\label{fig:distribute4logistic}
\end{figure}

\begin{figure}[!htb]
\centering
\begin{minipage}[t]{\imagewidth}
\centering
\includegraphics[width=\imagewidth]{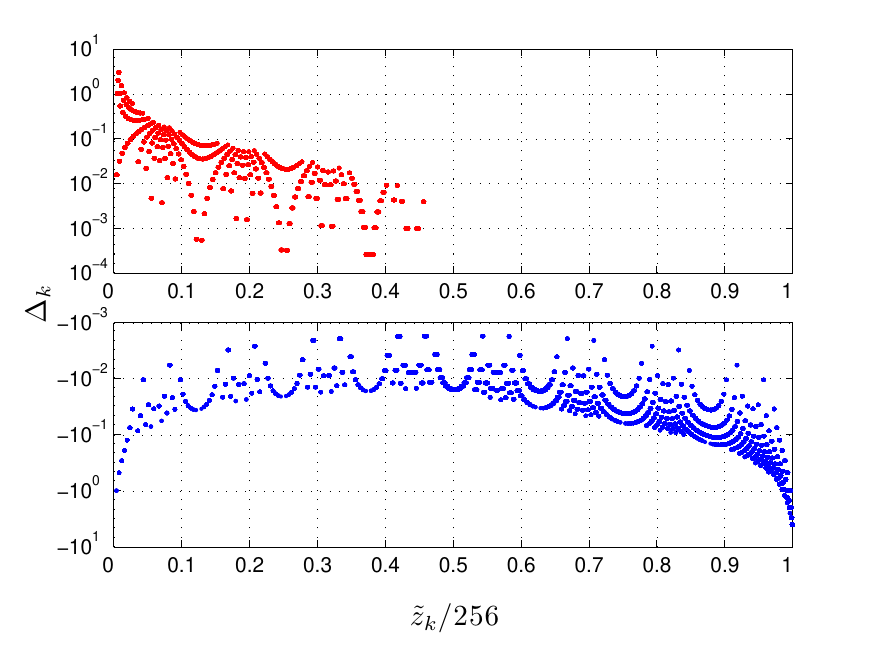}
\end{minipage}
\caption{Distribution of  $\{(\tilde{z}_{k}/256, \Delta_k)\}_{k=1}^{10^5}$ when $z_0=0.226$.}
\label{fig:z0z1}
\end{figure}

\begin{figure}[!htb]
\centering
\begin{minipage}[t]{\imagewidth}
\centering
\includegraphics[width=\imagewidth]{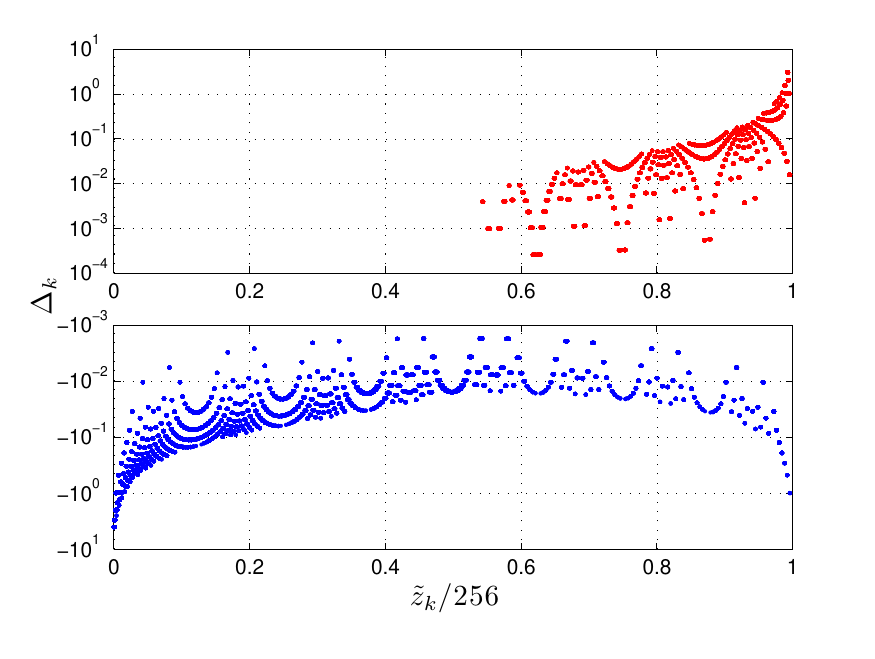}
\end{minipage}
\caption{The version of the distribution shown in Fig.~\ref{fig:z0z1} when only $\tilde{z}_{k}$ is wrong: $\tilde{z}_{k}$ is replaced by $\tilde{z}_{k}+1$.}
\label{fig:z0+1z1}
\end{figure}

\begin{figure}[!htb]
\centering
\begin{minipage}[t]{\imagewidth}
\centering
\includegraphics[width=\imagewidth]{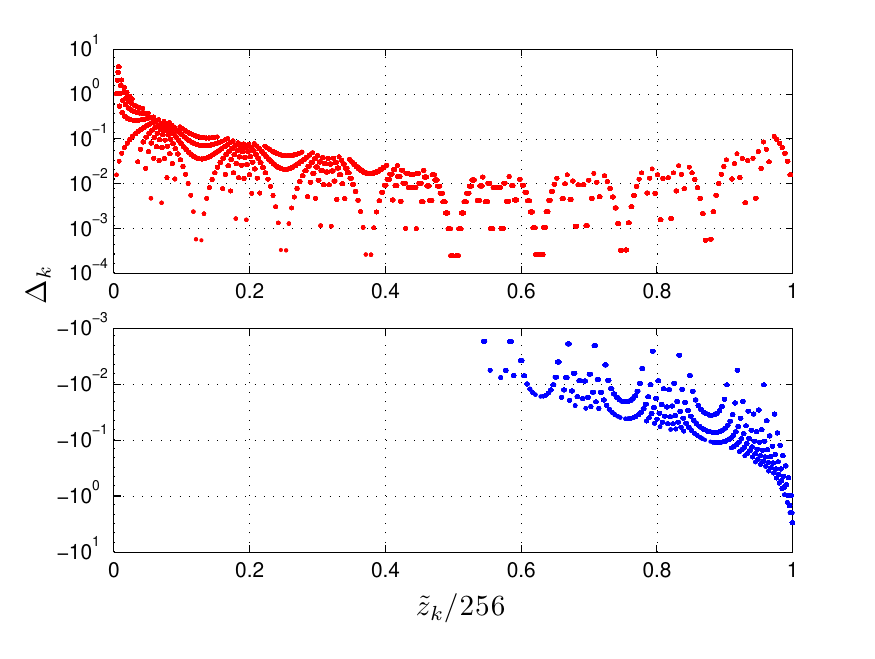}
\end{minipage}
\caption{The version of the distribution shown in Fig.~\ref{fig:z0z1} when only $\tilde{z}_{k+1}$
is wrong: $\tilde{z}_{k+1}$ is replaced by $\tilde{z}_{k+1}+1$.
}
\label{fig:z0z1+1}
\end{figure}

\begin{figure}[!htb]
\centering
\begin{minipage}[t]{\imagewidth}
\centering
\includegraphics[width=\imagewidth]{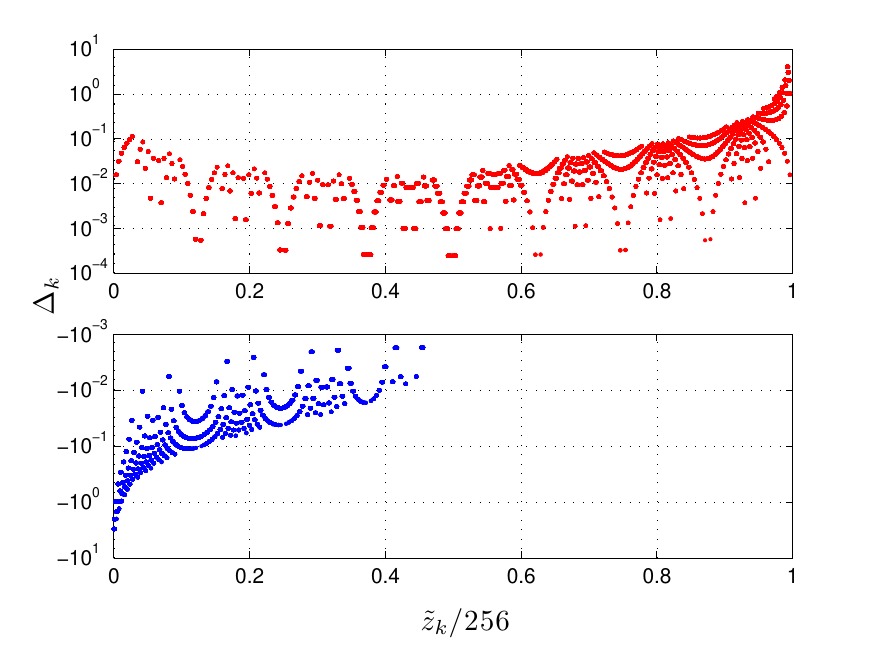}
\end{minipage}
\caption{The version of the distribution shown in Fig.~\ref{fig:z0z1} when $\tilde{z}_{k+1}$
and $\tilde{z}_{k}$ are birth wrong: $\tilde{z}_{k}$ and $\tilde{z}_{k+1}$ are replaced by $\tilde{z}_{k}+1$ and $\tilde{z}_{k+1}+1$ respectively.}
\label{fig:z0+1z1+1}
\end{figure}

\begin{table*}[!htp]
\center\caption{The distribution of $\Delta_k$ corresponding to the initial state $z_0=0.226$.}
\begin{tabular}{c|c|c|c|c|c}
  \hline  $z_0$  & $(-\infty, -10)$     &  $[-10^{i+1}, -10^{i})_{i=0}^{-4}$       & $[-10^{-4}, 10^{-4})$ &    $[10^{i}, 10^{i+1})_{i=-4}^0$& $[10,+\infty )$\\
  \hline  0.326  & 0.0195 & 0.0474    0.1435    0.3216    0.0897   0 & 0.0267 &   0.0150    0.0672    0.1588    0.0709    0.0196        & 0.0202       \\
  \hline  0.761  & 0.0204 & 0.0481    0.1431    0.3177    0.0899   0 & 0.0274 &   0.0143    0.0678    0.1610    0.0709    0.0197        & 0.0198       \\
  \hline  0.539  & 0.0203 & 0.0469    0.1437    0.3222    0.0905   0 & 0.0270 &   0.0141    0.0672    0.1581    0.0705    0.0192        & 0.0203       \\
  \hline  0.487  & 0.0200 & 0.0480    0.1412    0.3230    0.0905   0 & 0.0265 &   0.0146    0.0660    0.1593    0.0716    0.0193        & 0.0201       \\
  \hline  0.194  & 0.0206 & 0.0491    0.1419    0.3199    0.0879   0 & 0.0276 &   0.0150    0.0666    0.1599    0.0719    0.0192        & 0.0204       \\
  \hline  0.875  & 0.0210 & 0.0484    0.1430    0.3201    0.0893   0 & 0.0271 &   0.0152    0.0662    0.1603    0.0698    0.0198        & 0.0199       \\
  \hline  0.942  & 0.0211 & 0.0483    0.1417    0.3217    0.0877   0 & 0.0270 &   0.0154    0.0663    0.1597    0.0712    0.0194        & 0.0205       \\
  \hline  0.293  & 0.0205 & 0.0479    0.1426    0.3201    0.0905   0 & 0.0267 &   0.0147    0.0662    0.1597    0.0725    0.0187        & 0.0199       \\
  \hline
\end{tabular}
\label{table:distributionerror}
\end{table*}

Based on the above two points, some information of the secret key, $\{\textit{Xkey}(i)\}_{i=0}^3$, can be obtained with the following steps.
\begin{enumerate}
\item Obtain $\mymatrix{I}_{\key}$ via
\begin{equation*}
\mymatrix{I}_{\key}= \VD(\HD(\mymatrix{I}))\oplus \mymatrix{I}';
\end{equation*}

\item Search the value of $\{\textit{Xkey}(i)\}_{i=0}^3$, and get the corresponding version of $\mymatrix{I}_{\Xkey}^*$;

\item Generate estimated version of $\{\textit{CKSB}(i,j)\}_{0\leq i\leq M-1 \atop 0\leq j\leq N-1}$,
$\{\textit{CKSB}^*(i,j)\}_{0\leq i\leq M-1 \atop 0\leq j\leq N-1}$, from the blue channel of
$\mymatrix{I}^*_{\CKS}=\mymatrix{I}_{\key}\oplus\mymatrix{I}_{\Xkey}^*$ (See Eq.~(\ref{eq:equivalentkey})).

\item Calculate the distribution of $\{\Delta_k\}_{k=0}^{MN-2}$, and output the search value if the distribution match with the expected one (Refer to Table~\ref{table:distributionerror}).
\end{enumerate}

Once $\{\textit{Xkey}(i)\}_{i=0}^3$ are determined, one can obtain $K=\textit{Xkey}(2)\bmod 256$, $L=\textit{Xkey}(3)+n\cdot 256$, $n=1\sim 4$, $x_0\in [\textit{Xkey}(0)/256\cdot(2\pi), (\textit{Xkey}(0)+1)/256\cdot(2\pi))$
, $y_0\in [\textit{Xkey}(1)/256\cdot(2\pi), (\textit{Xkey}(1)+1)/256\cdot(2\pi))$.

Obviously, the complexity of the whole attack is $O(2^{4\cdot 8}\cdot L)$, where $L$ is the number of plain-bytes to be calculated.
Considering $L=256$ is enough for counting the distribution of $\{\Delta_k\}_{k=0}^{MN-2}$, the complexity of the attack is $O(2^{40})$. Note that the complexity of
the attack can be reduced much by dividing the attack into the following two stages: 1) search for the possible values of $\{\textit{Xkey}(i)\}_{i=0}^3$ of regular interval with a weaker matching condition;
2) verify the left possible values of $\{\textit{Xkey}(i)\}_{i=0}^3$ with a stronger matching condition. In our experiments, only the even possible values of $\{\textit{Xkey}(i)\}_{i=0}^3$
are searched, the matching condition is set as
\begin{equation*}
\#\left( \{k \mid |\Delta_k|<1/16\} \right)>(MN/2),
\end{equation*}
where $\#(\cdot)$ denotes cardinality of a set.
In this case, the complexity is reduced to $O(2^{36})$. For a PC with CPU of 2.83GHz and RAM of 2.98GB, the attack can be
completed within eleven minutes. A number of experiments were made to verify the correctness of the
above attack. With the secret key $(x_0, y_0, K,
L)=(3.98235562892545,$ $ 1.34536356538912,$ $ 108.54365761256745,$ $
110)$, the out of the first stage of the attack is $\{160,54,108,108\}$, $\{162,54,108,110\}$. Then, the 16 possible neighboring values of each
set are verified further, the distributions of $\Delta_k$ correspondig for the second set are shown in Table~\ref{table:result}. Obviously, the
data shown in third row of Table~\ref{table:result} is closed to the data shown in Table~\ref{table:distributionerror} most. So,  $\{\textit{Xkey}(i)\}_{i=0}^3=\{162, 54, 109, 110\}$.
One can further obtain that $K=109\bmod 256$, $L=110+n\cdot 256$, $n=1\sim 3$, $x_0\in [81/64\cdot \pi, 163/128\cdot \pi)$, $y_0\in [27/64\cdot\pi, 55/128\cdot\pi)$.

\begin{table*}[!htp]
\center\caption{The distribution of $\Delta_k$ corresponding to 16 possible neightouring values of \{162,54,108,110\}.}
\begin{tabular}{c|c|c|c}
  \hline  $index$      &  $[-10^{i+1}, -10^{i})_{i=0}^{-4}$       & $[-10^{-4}, 10^{-4})$ &    $[10^{i}, 10^{i+1})_{i=-4}^0$\\
  \hline   1  &  0.0629 0.1614 0.2716 0.0629 0.0000 &0.0157 &0.0078 0.1023 0.1811 0.1023 0.0275  \\
  \hline   2  &  0.0515 0.1666 0.2698 0.0238 0.0000 &0.0357 &0.0119 0.0753 0.2023 0.1230 0.0357  \\
  \hline   3  &  0.0476 0.1468 0.3373 0.0714 0.0000 &0.0357 &0.0119 0.0992 0.1626 0.0714 0.0119  \\\hline
  \hline   4  &  0.0480 0.1600 0.3240 0.0480 0.0000 &0.0240 &0.0080 0.0720 0.2000 0.0800 0.0320  \\
  \hline   5  &  0.0557 0.1474 0.3187 0.0557 0.0000 &0.0438 &0.0119 0.0876 0.1434 0.1235 0.0079  \\
  \hline   6  &  0.0478 0.1673 0.2709 0.0557 0.0000 &0.0239 &0.0039 0.0876 0.2151 0.0956 0.0278  \\
  \hline   7  &  0.0632 0.1660 0.2885 0.0553 0.0000 &0.0197 &0.0079 0.0869 0.1699 0.1067 0.0316  \\
  \hline   8  &  0.0434 0.1778 0.2529 0.0395 0.0000 &0.0316 &0.0079 0.0869 0.2252 0.0909 0.0395  \\
  \hline   9  &  0.0634 0.1587 0.3055 0.0674 0.0000 &0.0317 &0.0079 0.0634 0.1587 0.1230 0.0158  \\
  \hline   10 &  0.0515 0.1666 0.2698 0.0714 0.0000 &0.0198 &0.0000 0.0753 0.2182 0.1031 0.0198  \\
  \hline   11 &  0.0634 0.1507 0.3492 0.0714 0.0000 &0.0079 &0.0079 0.0753 0.1468 0.0873 0.0357   \\
  \hline   12 &  0.0396 0.1507 0.3253 0.0595 0.0000 &0.0277 &0.0079 0.0873 0.1904 0.0793 0.0277   \\
  \hline   13 &  0.0632 0.1581 0.2964 0.0790 0.0000 &0.0079 &0.0039 0.0909 0.1620 0.1027 0.0316   \\
  \hline   14 &  0.0478 0.1513 0.3067 0.0438 0.0000 &0.0358 &0.0079 0.0756 0.1713 0.1314 0.0239   \\
  \hline   15 &  0.0553 0.1699 0.2885 0.0830 0.0000 &0.0276 &0.0039 0.0750 0.1818 0.0909 0.0197   \\
  \hline   16 &  0.0517 0.1713 0.2868 0.0637 0.0000 &0.0239 &0.0000 0.0597 0.2071 0.1075 0.0239   \\
  \hline
\end{tabular}
\label{table:result}
\end{table*}

\section{Conclusion}

In this paper, the security of a new image encryption scheme based
on logistic map is re-analyzed in detail. Although it was reported that equivalent
secret key of the encryption scheme can be reconstructed with only one pair of known-plaintext
and the corresponding ciphertext, this paper further finds that
the scope of all the sub-keys can be narrowed further with relatively low computation under the same condition.
The results reported in this paper will help study security of other image encryption schemes
based on logistic map.

\begin{acknowledgements}
This research was supported by the National Natural Science Foundation of China (No. 61100216) and the
Alexander von Humboldt Foundation of Germany.
\end{acknowledgements}

\bibliographystyle{spmpsci}
\bibliography{Pareek2}
\end{document}